\documentclass[10pt,a4paper]{article}

\usepackage[T1]{fontenc}
\usepackage[utf8]{inputenc}
\usepackage[english]{babel}
\usepackage{lmodern}
\usepackage{microtype}

\usepackage{geometry}
\usepackage{amsmath,amssymb,amsthm,mathtools}
\usepackage{bm}

\newtheorem{lemma}{Lemma}
\newtheorem{proposition}{Proposition}
\newtheorem{corollary}{Corollary}
\theoremstyle{definition}
\newtheorem{definition}{Definition}
\theoremstyle{remark}
\newtheorem{remark}{Remark}

\DeclareMathOperator{\CT}{CT}

\usepackage{enumitem}
\usepackage{hyperref}
\usepackage[nameinlink,noabbrev]{cleveref}
\hypersetup{
  colorlinks=true,
  linkcolor=blue,
  citecolor=blue,
  urlcolor=blue
}

\newcommand{\1}{\mathbf{1}}
\newcommand{\PP}{\mathbb{P}}
\newcommand{\EE}{\mathbb{E}}

\title{Pareto-type finite-block optimality for source codes:\
a constrained Markov example}
\author{Stefano Della Fiore\\
\textit{Department of Information Engineering} \\
\textit{University of Brescia}\\
Brescia, Italy \\
stefano.dellafiore@unibs.it
}
\date{}

\begin{document}
\maketitle
 
\begin{abstract}
We study a Pareto-type notion of finite-block optimality for injective source codes, where two codes are compared through the full sequence of expected block lengths. As a concrete and fully analyzable test case, we revisit the four-symbol constrained Markov source introduced by Dalai and Leonardi in their ``meaningful example'' on constrained-source decodability. For each admissible nonempty string $u=x_1^m\in\mathcal A\subset\mathcal X^+$, let
\[
K(u):=-\log_2 \PP(X_1^m=u)
\]
denote its information cost. We construct a canonical injective binary mapping $C:\mathcal A\to\{0,1\}^+$ by ordering admissible strings by increasing $K(u)$, then by length and lexicographic order, and assigning binary strings in shortlex order. For the length-$n$ block $X_1^n$ we prove
\[
\EE[|C(X_1)|]=\tfrac32,
\qquad
\EE[|C(X_1^n)|]<\tfrac32\,n\quad (n\ge 2).
\]
Moreover, for every fixed
\[
0<c<\frac{\sqrt2}{18\sqrt\pi}
\]
we have
\[
\EE[|C(X_1^n)|]\le \tfrac32\,n-\frac{c}{\sqrt n}
\]
for all sufficiently large $n$. Thus, for this source, the reversible Dalai--Leonardi code is not Pareto-optimal with respect to finite-block average length. The proof is based on an exact enumeration of admissible strings by information cost and on a shortlex gap identity implying that each cost class splits evenly between lengths $K(u)-1$ and $K(u)$. The example is simple, but it already exhibits the kind of finite-block Pareto comparison that seems natural for injective source coding under source constraints.
\end{abstract}

\section{Introduction}
Shannon's work on noiseless coding already included sources with memory and finite-state structure \cite{Shannon48}; see also the historical discussion in \cite{Verdu98}. In the classical source-independent setting, Kraft's inequality and McMillan's theorem lead to the standard lower bound of expected code length by entropy \cite{Kraft49,McMillan56,CoverThomas,Gallager68}. For constrained sources, however, one must distinguish between different notions of decodability. A binary assignment may fail to be uniquely decodable in the usual concatenation-based sense and yet still be perfectly invertible on the set of strings that the source can actually generate. Dalai and Leonardi made this point explicit in \cite{DalaiLeonardi}.

Once one adopts injectivity on the admissible source language as the basic requirement, a natural finite-block comparison problem emerges. Given two injective source codes $C$ and $\widetilde C$ for the same source, one may compare them through the entire sequence
\[
\bigl(\EE[|C(X_1^n)|]\bigr)_{n\ge1},
\]
and ask whether one code dominates the other in a Pareto-type sense: does it satisfy
\[
\EE[|C(X_1^n)|]\le \EE[|\widetilde C(X_1^n)|]\qquad\text{for every }n\ge1,
\]
with strict inequality for at least one blocklength? This viewpoint is especially natural in constrained-source coding, where reversibility depends on the source law and where finite-block effects need not be visible from entropy alone.

The four-symbol Markov source exhibited by Dalai and Leonardi in \cite[Sec.~II]{DalaiLeonardi} provides a particularly clean test case for this question. Their reversible, source-aware binary encoding has expected length exactly $3n/2$ on the first $n$ source symbols, while the block entropy equals $3n/2+1/2$. The example is elementary, but significant: it shows that source constraints can create finite-block gains that are invisible if one insists on a source-independent reading of unique decodability. It also suggests a sharper question: is the Dalai--Leonardi reversible code Pareto-optimal with respect to the sequence of finite-block average lengths, or can one build another injective source code that never performs worse and performs strictly better somewhere?

This question is also related in spirit to the literature on one-to-one codes, where injectivity is studied without requiring the full prefix or uniquely decodable structure used in the classical theory. A useful reference point is the lower bound of Alon and Orlitsky on the expected length of one-to-one codes \cite{AlonOrlitsky94}. Our setting is different because the admissible source language is constrained by a Markov law, but the same general perspective remains relevant: once injectivity rather than source-independent unique decodability becomes the baseline requirement, one can ask for finer finite-block comparisons between reversible constructions.

In this paper we answer the above question for the Dalai--Leonardi source. We define a canonical injective block code on the whole admissible language $\mathcal A$: admissible strings are ordered by increasing information cost
\[
K(u):=-\log_2 \PP(X_1^m=u),
\]
then by length and lexicographic order, and are mapped to nonempty binary strings in shortlex order. We then analyze the expected length of this code on source blocks of length $n$.

Our first result is the uniform comparison
\[
\EE[|C(X_1^n)|]\le \frac32 n\qquad (n\ge1),
\]
with equality only at $n=1$. Thus the code never performs worse, in average length, than the reversible construction of \cite{DalaiLeonardi}, and it is strictly better at every blocklength $n\ge2$. More precisely, we prove the quantitative asymptotic gain
\[
\EE[|C(X_1^n)|]\le \frac32 n-\frac{c}{\sqrt n}
\]
for every fixed $c<\sqrt2/(18\sqrt\pi)$ and all sufficiently large $n$. In particular, the Dalai--Leonardi benchmark is not Pareto-optimal for this source in the finite-block sense above.

The paper is self-contained. We begin by recalling the source and its basic dyadic structure, then turn to the combinatorial organization of admissible strings by information cost. The core of the argument is a shortlex gap identity showing that each cost class splits exactly into two equal parts, one receiving code length $K(u)-1$ and the other receiving code length $K(u)$. This leads to a one-bit saving with probability at least $1/2$ at every blocklength, and with probability strictly larger than $1/2$ for every blocklength $n\ge2$. A refined central estimate then gives the explicit $n^{-1/2}$ improvement stated above.

\section{The Markov source}
We now specialize to the four-symbol constrained Markov source that appears as the ``meaningful example'' in \cite[Sec.~II]{DalaiLeonardi}. Let $\mathcal X=\{A,B,C,D\}$ and consider the first-order Markov source $(X_i)_{i\ge 1}$ with
\[
\PP(X_1=x)=\frac14,\qquad x\in\mathcal X,
\]
and transitions
\[
A\to\{A,C\}\ (1/2,1/2),\qquad
B\to\{B,D\}\ (1/2,1/2),\qquad
C,D\to\mathcal X\ \text{uniform}.
\]
A finite string $x_1^n$ is \emph{admissible} if all transitions $x_i\to x_{i+1}$ are allowed. Let $\mathcal A\subset\mathcal X^+$
be the set of admissible nonempty finite strings.

We begin with a simple structural fact that will be used repeatedly in the sequel.

\begin{lemma}[Uniform marginals]
For every $i\ge 1$ and every $x\in\mathcal X$, $\PP(X_i=x)=1/4$.
\end{lemma}
\begin{proof}
The uniform distribution $\pi=(1/4,1/4,1/4,1/4)$ is invariant: from $A$ (resp.\ $B$) mass $1/4$ splits equally to $A,C$ (resp.\ $B,D$),
and from $C,D$ mass $1/2$ is sent uniformly to all four symbols. Hence $\pi P=\pi$. Since $X_1\sim \pi$, induction gives $X_i\sim \pi$ for all $i$.
\end{proof}

The next step is to record the dyadic form of all admissible probabilities and the resulting information-cost formula.

From $A,B$ there are two equiprobable next symbols (cost $1$ bit), from $C,D$ there are four equiprobable next symbols (cost $2$ bits). Hence every admissible block has dyadic probability.

\begin{definition}[Information cost]
For an admissible block $x_1^n$ define
\[
K(x_1^n):=-\log_2 \PP(X_1^n=x_1^n)\in\mathbb Z_{\ge 0}.
\]
\end{definition}

For the present source, letting
\[
M(x_1^n):=\#\{\, i\in\{1,\dots,n-1\}: x_i\in\{C,D\}\,\},
\]
we have
\begin{equation}\label{eq:K-formula}
K(x_1^n)=2+\sum_{i=1}^{n-1}\Bigl(\1_{\{x_i\in\{A,B\}\}}+2 \cdot \1_{\{x_i\in\{C,D\}\}}\Bigr)
=(n+1)+M(x_1^n).
\end{equation}

\section{The shortlex code on admissible strings}
We now introduce the canonical injective code studied throughout the paper and collect the combinatorial facts that govern its behavior.

Let $\prec_{\mathrm{bin}}$ be the shortlex order on nonempty binary strings: first by length, then lexicographic. On admissible strings we use the order $\prec_{\mathcal X}$: first by increasing cost $K(\cdot)$, then by increasing length, then lexicographic over $A<B<C<D$.

\begin{definition}[Injective source code on $\mathcal A$]
A mapping $C:\mathcal A\to\{0,1\}^+$ is called \emph{injective} if it is one-to-one on $\mathcal A$.
This guarantees unique one-shot decoding on admissible source strings. We do not claim any concatenation-based
unique decodability in the classical sense.
\end{definition}

\begin{definition}[Shortlex source code]
Enumerate $\mathcal A$ as $(u_1,u_2,\dots)$ in the order $\prec_{\mathcal X}$.
Enumerate nonempty binary strings as $(b_1,b_2,\dots)$ in the order $\prec_{\mathrm{bin}}$.
Define $C(u_j)=b_j$ for $j\ge1$.
\end{definition}

By construction $C$ is injective on $\mathcal A$.

\begin{remark}[Numerical check for small $n$]
By computing the ranking and the shortlex assignment explicitly, one obtains for the Markov source defined in the previous section:
\[
\mathbb E[|C(X_1^2)|]=\frac{23}{8}=2.875<3=\frac32\cdot 2,\qquad
\mathbb E[|C(X_1^3)|]=\frac{71}{16}=4.4375<\frac92=\frac32\cdot 3.
\]
\end{remark}

We next count admissible strings by information cost. This will provide the basic input for understanding how the shortlex assignment behaves inside each cost class.

Let
\[
S_k:=\#\{u\in\mathcal A:\ K(u)=k\},\qquad k\ge 2.
\]

\begin{lemma}[Cost-class recurrence and closed form]\label{lem:Sk}
$S_2=S_3=4$ and for all $k\ge 4$,
\[
S_k=S_{k-1}+2S_{k-2}.
\]
Moreover,
\[
S_k=\frac{2^{k+1}+4(-1)^k}{3}\qquad (k\ge 2).
\]
\end{lemma}

\begin{proof}
Refine the count by the last symbol. Let $a_k$ be the number of admissible strings of cost $k$ ending in $\{A,B\}$, and $c_k$ those ending in $\{C,D\}$. Then $S_k=a_k+c_k$.

To form a cost-$k$ string ending in $\{A,B\}$:
\begin{itemize}[leftmargin=2em]
\item if the penultimate symbol is in $\{A,B\}$, the last step costs $1$ bit and the last symbol in $\{A,B\}$ is forced (A$\to$A, B$\to$B), so we get $a_{k-1}$ strings;
\item if the penultimate symbol is in $\{C,D\}$, the last step costs $2$ bits and the last symbol can be chosen in $\{A,B\}$ in $2$ ways, so we get $2c_{k-2}$ strings.
\end{itemize}
Hence $a_k=a_{k-1}+2c_{k-2}$. By the same reasoning (A$\to$C and B$\to$D forced; from $C,D$ two choices to land in $\{C,D\}$), also $c_k=a_{k-1}+2c_{k-2}$. With $a_2=c_2=2$, it follows by induction that $a_k=c_k$ for all $k$, and therefore
\[
S_k=2a_k,\qquad a_k=a_{k-1}+2a_{k-2}.
\]
Multiplying by $2$ gives $S_k=S_{k-1}+2S_{k-2}$ for $k\ge 4$.

The characteristic equation $\lambda^2=\lambda+2$ has roots $2$ and $-1$, so $S_k=\alpha 2^k+\beta(-1)^k$.
Matching $S_2=S_3=4$ gives $\alpha=2/3$ and $\beta=4/3$, i.e. $S_k=(2^{k+1}+4(-1)^k)/3$.
\end{proof}

Once the size of each cost class is known, one can compare it with the number of short binary strings available before that class begins. This is the content of the next identity.

Let
\[
T_{<k}:=\sum_{j=2}^{k-1} S_j
\]
be the number of admissible strings with cost $<k$.

\begin{lemma}[Exact gap identity]\label{lem:gap}
For every $k\ge 2$,
\[
(2^k-2)-T_{<k}=\frac{S_k}{2}.
\]
\end{lemma}

\begin{proof}
Set $G_k:=(2^k-2)-T_{<k}$ and $H_k:=G_k-\frac{S_k}{2}$.
Since $T_{<k+1}=T_{<k}+S_k$,
\[
H_{k+1}-H_k
=2^k-S_k-\frac{S_{k+1}-S_k}{2}.
\]
Using the closed form from \Cref{lem:Sk},
\[
S_k=\frac{2^{k+1}+4(-1)^k}{3},\qquad
S_{k+1}=\frac{2^{k+2}-4(-1)^k}{3},
\]
we obtain $H_{k+1}-H_k=0$. Thus $(H_k)$ is constant for $k\ge2$.
For $k=2$, $T_{<2}=0$ and $S_2=4$, hence $H_2=(2^2-2)-0-4/2=0$. Therefore $H_k=0$ for all $k\ge2$.
\end{proof}

We can now identify the only two possible code lengths inside a fixed cost class and the exact proportion in which they occur.

\begin{proposition}[Length dichotomy]\label{prop:dichotomy}
For any admissible string $u$ with $K(u)=k$,
\[
|C(u)|\in\{k-1,k\}.
\]
Moreover, within each cost class $k$, exactly $S_k/2$ strings receive length $k-1$ and exactly $S_k/2$ receive length $k$.
\end{proposition}

\begin{proof}
Before cost class $k$ starts, exactly $T_{<k}$ admissible strings have been encoded, hence the first $T_{<k}$ binary strings in shortlex have been used.

\smallskip
\noindent\emph{(i) All codewords of length $\le k-2$ are already used.}
There are $2^{k-1}-2$ nonempty binary strings of length $\le k-2$.
By \Cref{lem:gap},
\[
T_{<k}=2^k-2-\frac{S_k}{2}.
\]
Using \Cref{lem:Sk}, for every $k\ge2$ we have $S_k\le 2^k$ (indeed $(2^{k+1}+4)/3\le 2^k$ since $4\le 2^k$).
Thus $T_{<k}\ge 2^{k-1}-2$, so all binary strings of length $\le k-2$ are among the first $T_{<k}$ and are already used.

\smallskip
\noindent\emph{(ii) No codeword of length $\ge k$ is used before class $k$.}
Again by \Cref{lem:gap}, $T_{<k}=2^k-2-\frac{S_k}{2}<2^k-2$, so not all binary strings of length $\le k-1$
have been exhausted; hence no length-$k$ word can have appeared yet in shortlex.

\smallskip
Therefore, when class $k$ begins, the remaining unused codewords of length $\le k-1$ are precisely some of the length-$(k-1)$ codewords, and their number is
\[
(2^k-2)-T_{<k}=\frac{S_k}{2}
\]
by \Cref{lem:gap}. Assigning codewords in shortlex inside the class, exactly $S_k/2$ strings get length $k-1$ and the remaining $S_k/2$ strings get the next available length, namely $k$.
\end{proof}

\section{A one-bit saving bound}
We now turn to the probabilistic part of the argument. The goal of this section is to show that the shortlex construction saves one bit with probability at least $1/2$ at every blocklength, and with strictly larger probability as soon as $n\ge2$.

Fix $n\ge 1$ and define
\[
K_n:=K(X_1^n),\qquad L_n:=|C(X_1^n)|.
\]
By \Cref{prop:dichotomy}, $L_n\in\{K_n-1,K_n\}$. Define the saving indicator
\[
I_n:=\1\{L_n=K_n-1\}.
\]
Then $L_n=K_n-I_n$ and therefore
\begin{equation}\label{eq:Ln-decomp}
\EE[L_n]=\EE[K_n]-\PP(I_n=1).
\end{equation}
For this source,
\[
\EE[K_n]=H(X_1^n)=\frac32 n+\frac12,
\]
so it remains to lower bound the saving probability $\PP(I_n=1)$.

A first simplification is that the random cost parameter can be encoded by a binomial statistic.

Let
\[
X:=M(X_1^n)=\#\{i\in\{1,\dots,n-1\}:X_i\in\{C,D\}\}.
\]
Define $Y_i:=\1\{X_i\in\{C,D\}\}$. For every $x\in\mathcal X$,
\[
\PP(X_{i+1}\in\{C,D\}\mid X_i=x)=\frac12,
\]
hence $\PP(Y_{i+1}=1\mid X_1,\dots,X_i)=1/2$ and $(Y_i)_{i\ge1}$ is i.i.d.\ Bernoulli$(1/2)$. Therefore
\[
X=\sum_{i=1}^{n-1}Y_i\sim \mathrm{Bin}(n-1,\tfrac12),
\qquad
p_n(x):=\PP(X=x)=\binom{n-1}{x}2^{-(n-1)}.
\]
By \eqref{eq:K-formula},
\[
K_n=(n+1)+X.
\]

To translate this observation into a statement about code lengths, we need to count how many admissible strings have a given length and cost.

For integers $k\ge2$ and $\ell\ge1$, define
\[
N(k,\ell):=\#\{x_1^\ell\in\mathcal A:K(x_1^\ell)=k\}.
\]

\begin{lemma}[Closed form for $N(k,\ell)$]\label{lem:Nkl}
For any $\ell\ge1$ and any $j\in\{0,\dots,\ell-1\}$, the number of admissible strings $x_1^\ell$ with
\[
\#\{i\in\{1,\dots,\ell-1\}:x_i\in\{C,D\}\}=j
\]
equals
\[
4\binom{\ell-1}{j}2^j.
\]
Consequently,
\begin{equation}\label{eq:Nkl}
N(k,\ell)=4\binom{\ell-1}{k-\ell-1}2^{k-\ell-1},
\end{equation}
with the convention that this is $0$ when $k-\ell-1\notin\{0,\dots,\ell-1\}$.
\end{lemma}

\begin{proof}
Fix $\ell\ge1$ and a subset $J\subset\{1,\dots,\ell-1\}$ with $|J|=j$. We count admissible strings whose indices with $x_i\in\{C,D\}$ are exactly $J$.

First choose the class of $x_1$ (either $\{A,B\}$ or $\{C,D\}$): $2$ choices. For each $i\in\{1,\dots,\ell-2\}$:
\begin{itemize}[leftmargin=2em]
\item if $i\notin J$, then $x_i\in\{A,B\}$ and once the class of $x_{i+1}$ is fixed (by whether $i+1\in J$), the transition forces $x_{i+1}$ inside that class;
\item if $i\in J$, then $x_i\in\{C,D\}$ and the next class is again fixed by membership of $i+1$ in $J$, but from $C$ or $D$ there are exactly $2$ admissible choices inside the prescribed class.
\end{itemize}
For the final step $i=\ell-1$, if $\ell-1\notin J$ there are $2$ choices for $x_\ell$ inside the forced class, and if $\ell-1\in J$ there are $4$ choices. In either case, the total number of strings compatible with a fixed $J$ equals $4\cdot 2^j$. Summing over the $\binom{\ell-1}{j}$ possible subsets $J$ gives the first claim.

Finally, since $K(x_1^\ell)=(\ell+1)+j$ by \eqref{eq:K-formula}, we have $j=k-\ell-1$, which yields \eqref{eq:Nkl}.
\end{proof}

For later use, write
\[
U(m,j):=\binom{m-j}{j}2^j,
\qquad
U(m,j)=0\ \text{if } j<0 \text{ or } 2j>m.
\]
Then \eqref{eq:Nkl} can be rewritten as
\[
N(k,\ell)=4U(k-2,k-\ell-1).
\]

We now isolate the number of short codewords that are still available when the shortlex order reaches the length-$n$ slice of a fixed cost class.

Fix $n\ge2$ and a cost class $k$. The number of cost-$k$ strings of length $<n$ is
\[
W_{<n}(k):=\sum_{\ell=1}^{n-1}N(k,\ell).
\]
By \Cref{prop:dichotomy}, within cost class $k$ there are exactly $S_k/2$ available short codewords of length $k-1$. Therefore, when the encoding reaches the length-$n$ slice inside cost class $k$, the number of remaining short codewords equals
\begin{equation}\label{eq:Vnk}
V_n(k):=\frac{S_k}{2}-W_{<n}(k).
\end{equation}

The next lemma packages this observation into an explicit conditional probability formula.

\begin{lemma}[Conditional saving formula]\label{lem:cond}
Let $n\ge2$ and $x\in\{0,\dots,n-1\}$, and set $k:=n+1+x$. Define
\[
g_n(x):=\PP(I_n=1\mid X=x).
\]
Then
\begin{equation}\label{eq:gnV}
g_n(x)=\min\!\left\{1,\max\!\left\{0,\frac{V_n(k)}{N(k,n)}\right\}\right\}.
\end{equation}
Consequently,
\begin{equation}\label{eq:Pin-sum}
\PP(I_n=1)=\sum_{x=0}^{n-1}p_n(x)g_n(x),
\qquad
p_n(x)=\binom{n-1}{x}2^{-(n-1)}.
\end{equation}
\end{lemma}

\begin{proof}
Condition on $X=x$. Then $K_n=(n+1)+x=k$. All admissible strings in this event form exactly the length-$n$ slice of cost $k$, whose cardinality is $N(k,n)$. Within the whole cost class $k$, exactly $S_k/2$ strings receive short codewords (length $k-1$). Since the order $\prec_{\mathcal X}$ lists shorter source lengths first, by the time we begin encoding the length-$n$ slice we have already encoded exactly the $W_{<n}(k)$ strings of lengths $\ell<n$ in the same cost class. Hence the number of short codewords still available for the length-$n$ slice is $V_n(k)$ as in \eqref{eq:Vnk}.

If $V_n(k)\le0$, no short codeword remains, so $g_n(x)=0$. If $V_n(k)\ge N(k,n)$, there are enough short codewords to cover the entire slice, so $g_n(x)=1$. Otherwise $0<V_n(k)<N(k,n)$, and exactly $V_n(k)$ out of the $N(k,n)$ strings in the slice get a short codeword, so $g_n(x)=V_n(k)/N(k,n)$. This is exactly \eqref{eq:gnV}. Averaging over $X\sim\mathrm{Bin}(n-1,1/2)$ gives \eqref{eq:Pin-sum}.
\end{proof}

The rest of the proof separates two regimes: a right-tail regime where all strings are saved, and a central regime where one needs a more delicate oscillatory estimate.

\begin{lemma}[Tail saturation]\label{lem:tail}
Let $n\ge2$ and let $x\in\{0,\dots,n-1\}$. If
\[
x\ge \frac{n+1}{2},
\]
then
\[
g_n(x)=1.
\]
\end{lemma}

\begin{proof}
Fix $n\ge2$ and $x\ge (n+1)/2$, and set
\[
y:=n-1-x,
\qquad
k:=n+1+x,
\qquad
m:=k-2=n+x-1=x+2y.
\]
Using $N(k,\ell)=4U(k-2,k-\ell-1)$, we obtain
\[
N(k,n)=4U(m,x),
\qquad
W_{<n}(k)=4\sum_{j\ge x+1}U(m,j),
\qquad
S_k=4\sum_{j\ge0}U(m,j).
\]
Hence
\[
V_n(k)-N(k,n)
=2\left(\sum_{j=0}^{x-1}U(m,j)-\sum_{j\ge x+1}U(m,j)\right).
\]
Therefore it suffices to prove
\begin{equation}\label{eq:tail-target-v2}
\sum_{j=0}^{x-1}U(m,j)\ge \sum_{j\ge x+1}U(m,j).
\end{equation}

Since $m=x+2y$, the support of $j\mapsto U(m,j)$ is
\[
0\le j\le \Bigl\lfloor\frac m2\Bigr\rfloor=x+\Bigl\lfloor\frac y2\Bigr\rfloor,
\]
so the right-hand side of \eqref{eq:tail-target-v2} equals
\[
\sum_{r=1}^{\lfloor y/2\rfloor}U(m,x+r).
\]
We compare these terms with the left tail via
\[
R_r:=\frac{U(m,x-1-r)}{U(m,x+r)}
\qquad
\Bigl(0\le r\le \Bigl\lfloor\frac y2\Bigr\rfloor\Bigr).
\]
First,
\[
R_0=\frac{U(m,x-1)}{U(m,x)}=\frac{x(x+y+1)}{2(y+2)(y+1)}>1,
\]
because $x\ge y+2$.
Next, for $0\le r<\lfloor y/2\rfloor$, a direct computation gives
\[
\frac{R_{r+1}}{R_r}
=
\frac{(x+y-r)(x-r-1)(x+r+1)(x+y+r+2)}
{4(y-2r)(y-2r-1)(y+2r+3)(y+2r+4)}.
\]
The right-hand side is increasing in $x$, so it is bounded below by its value at the smallest admissible value $x=y+2$. Therefore
\[
\frac{R_{r+1}}{R_r}
\ge
\frac{(2y+2-r)(y-r+1)(y+r+3)(2y+r+4)}
{4(y-2r)(y-2r-1)(y+2r+3)(y+2r+4)}.
\]
Split this lower bound as
\[
\frac{(y-r+1)(y+r+3)}{(y-2r)(y+2r+4)}
\cdot
\frac{(2y+2-r)(2y+r+4)}{4(y-2r-1)(y+2r+3)}.
\]
Both factors are strictly larger than $1$, because
\[
(y-r+1)(y+r+3)-(y-2r)(y+2r+4)=3(r+1)^2>0,
\]
and
\[
(2y+2-r)(2y+r+4)-4(y-2r-1)(y+2r+3)=15r^2+30r+4y+20>0.
\]
Hence $R_{r+1}>R_r$, and therefore $R_r\ge R_0>1$ for every admissible $r$. Thus
\[
U(m,x-1-r)\ge U(m,x+r)
\qquad
\Bigl(0\le r\le \Bigl\lfloor\frac y2\Bigr\rfloor\Bigr).
\]
Summing over $r$ gives
\[
\sum_{j=0}^{x-1}U(m,j)
\ge
\sum_{r=0}^{\lfloor y/2\rfloor}U(m,x-1-r)
\ge
\sum_{r=1}^{\lfloor y/2\rfloor}U(m,x+r),
\]
which is exactly \eqref{eq:tail-target-v2}. Therefore $V_n(k)\ge N(k,n)$ and \Cref{lem:cond} yields $g_n(x)=1$.
\end{proof}

To handle the remaining values near the center, we introduce a family of alternating sums that captures the residual imbalance inside the relevant cost classes.

For $t\ge1$, define
\[
C_t:=U(3t,t)=\binom{2t}{t}2^t,
\]
and
\[
D_t:=\sum_{j=0}^{t-1}U(3t,j)-\sum_{j\ge t+1}U(3t,j).
\]

\begin{lemma}[Alternating-sum estimates]\label{lem:central}
For all $t\ge1$,
\[
D_t=C_{t-1}-C_{t-2}+C_{t-3}-\cdots+(-1)^{t-1}C_0,
\]
and in particular
\[
0<D_t<\frac12 C_t.
\]
Moreover, if
\[
B_t:=\sum_{j=0}^{t-1}U(3t-1,j)-\sum_{j\ge t+1}U(3t-1,j),
\]
then
\[
B_t=\frac12 C_t-D_t,
\qquad
0<B_t<\frac12 C_t.
\]
Finally,
\[
D(z):=\sum_{t\ge1}D_tz^t=\frac{z}{(1+z)\sqrt{1-8z}}.
\]
\end{lemma}

\begin{proof}
Set
\[
P_t:=\sum_{j=0}^{t-1}U(3t,j),
\qquad
T_t:=\sum_{j\ge0}U(3t,j),
\]
so that
\begin{equation}\label{eq:D-decomp-v2}
D_t=2P_t+C_t-T_t.
\end{equation}
First, $T_t=S_{3t+2}/4$, so by \Cref{lem:Sk},
\[
T_t=\frac{2^{3t+3}+4(-1)^t}{12}=\frac{2\cdot 8^t+(-1)^t}{3},
\qquad
T(z):=\sum_{t\ge1}T_tz^t=\frac13\left(\frac{16z}{1-8z}-\frac{z}{1+z}\right).
\]
Next,
\[
C(z):=\sum_{t\ge1}C_tz^t=\sum_{t\ge1}\binom{2t}{t}(2z)^t=\frac{1}{\sqrt{1-8z}}-1.
\]

To compute $P(z):=\sum_{t\ge1}P_tz^t$, use constant terms:
\[
U(3t,j)=\CT_x\!\left((1+x)^{3t}\Bigl(\frac{2}{x(1+x)}\Bigr)^j\right).
\]
Summing first over $t$ and then over $j$ gives
\[
P(z)
=
\CT_x\left(
\frac{z\,x(1+x)^3}{(1-z(1+x)^3)\bigl(x-2z(1+x)^2\bigr)}
\right).
\]
To evaluate $P(z)$, we use a standard coefficient form of the Lagrange--B\"urmann inversion formula; see, for instance, \cite[Sec.~2.1]{GesselLagrange} (and also \cite[Sec.~4.1]{GesselLagrange} for its residue interpretation). Let $f=f(u)$ be the unique formal power series satisfying
\[
f=u\,\phi(f),
\qquad
\phi(x):=2z(1+x)^2,
\]
and let
\[
F(x):=\frac{z(1+x)^3}{1-z(1+x)^3}.
\]
Then Lagrange inversion gives, for every $n\ge 0$,
\[
[u^n]\frac{F(f(u))}{1-u\phi'(f(u))}
=
[x^n]\bigl(F(x)\phi(x)^n\bigr).
\]
Summing over $n\ge0$ and using
\[
\sum_{n\ge0}[x^n]\bigl(F(x)\phi(x)^n\bigr)
=
\CT_x\!\left(F(x)\sum_{n\ge0}\Bigl(\frac{\phi(x)}{x}\Bigr)^n\right)
=
\CT_x\!\left(\frac{x}{x-\phi(x)}F(x)\right),
\]
we obtain
\[
\sum_{n\ge0}[u^n]\frac{F(f(u))}{1-u\phi'(f(u))}
=
\CT_x\!\left(\frac{x}{x-\phi(x)}F(x)\right).
\]
Now set $u=1$ and write $X:=f(1)$. Since $f=u\phi(f)$, the series $X$ is characterized by
\[
X=\phi(X),
\qquad X(0)=0.
\]
Therefore
\[
\CT_x\!\left(\frac{x}{x-\phi(x)}F(x)\right)
=
\frac{F(X)}{1-\phi'(X)}.
\]
In our case this gives
\[
P(z)=\frac{F(X)}{1-\phi'(X)},
\qquad
X=\phi(X)=2z(1+X)^2.
\]
Solving the quadratic equation gives
\[
X(z)=\frac{1-4z-\sqrt{1-8z}}{4z},
\]
and substituting into the previous expression yields
\[
P(z)=\frac{1-\sqrt{1-8z}-2z}{2(1+z)(1-8z)}.
\]
Combining with \eqref{eq:D-decomp-v2},
\[
D(z)=2P(z)+C(z)-T(z)=\frac{z}{(1+z)\sqrt{1-8z}}.
\]
Since
\[
\frac1{\sqrt{1-8z}}=\sum_{t\ge0}C_tz^t,
\]
we obtain
\[
D(z)=z\Bigl(\sum_{s\ge0}(-1)^sz^s\Bigr)\Bigl(\sum_{r\ge0}C_rz^r\Bigr),
\]
and coefficient extraction gives
\[
D_t=C_{t-1}-C_{t-2}+C_{t-3}-\cdots+(-1)^{t-1}C_0.
\]
Grouping in pairs shows $D_t>0$ for all $t\ge1$. Also,
\[
\frac{C_t}{C_{t-1}}=2\frac{\binom{2t}{t}}{\binom{2t-2}{t-1}}=8-\frac4t>2,
\]
so $D_t<C_{t-1}<C_t/2$ for $t\ge2$, while $D_1=1<C_1/2=2$.

Finally, the same constant-term computation for $B_t$ gives
\[
B(z):=\sum_{t\ge1}B_tz^t
=
\frac12\left(\frac{1-z}{(1+z)\sqrt{1-8z}}-1\right)
=
\frac12\bigl(C(z)-2D(z)\bigr),
\]
whence $B_t=C_t/2-D_t$. Since $0<D_t<C_t/2$, we conclude that $0<B_t<C_t/2$.
\end{proof}

We are now ready to prove the main probabilistic estimate from which the average-length bounds follow.

\begin{lemma}[Key lemma]\label{lem:key}
For every $n\ge1$,
\[
\PP(I_n=1)\ge \frac12.
\]
Moreover, the inequality is strict for every $n\ge2$.
\end{lemma}

\begin{proof}
For $n=1$ we have $K_1=2$ and the four symbols are equiprobable; by \Cref{prop:dichotomy} exactly two receive length $1=K_1-1$, hence $\PP(I_1=1)=1/2$.

Assume now $n\ge2$. By \Cref{lem:cond},
\[
\PP(I_n=1)=\sum_{x=0}^{n-1}p_n(x)g_n(x),
\qquad
p_n(x)=\binom{n-1}{x}2^{-(n-1)}.
\]
We pair $x$ with $n-1-x$.

\smallskip
\noindent\emph{Odd case: $n=2t+1$ ($t\ge1$).}
For $x<t$, its partner $y=2t-x$ satisfies $y\ge t+1>(2t+2)/2$, hence $g_n(y)=1$ by \Cref{lem:tail}. For the central value $x=t$, we have $k=3t+2$,
\[
W_{<n}(k)=4\sum_{j\ge t+1}U(3t,j),
\qquad
\frac{S_k}{2}=2\sum_{j\ge0}U(3t,j),
\qquad
N(k,n)=4C_t.
\]
Therefore
\[
V_n(k)=2\left(C_t+D_t\right),
\]
so by \Cref{lem:cond} and \Cref{lem:central},
\[
g_n(t)=\frac{V_n(k)}{N(k,n)}=\frac12+\frac{D_t}{2C_t}\in\Bigl(\frac12,1\Bigr).
\]
Hence
\[
\PP(I_n=1)
\ge \sum_{x=0}^{t-1}p_n(x)+p_n(t)g_n(t)
=\frac12+p_n(t)\frac{D_t}{2C_t}
>\frac12.
\]

\smallskip
\noindent\emph{Even case: $n=2t$ ($t\ge1$).}
For $x\le t-2$, its partner $y=2t-1-x$ satisfies $y\ge t+1>(2t+1)/2$, hence $g_n(y)=1$ by \Cref{lem:tail}. Only the adjacent pair $(t-1,t)$ remains.

For $x=t$, we have $k=3t+1$, and \Cref{lem:central} gives
\[
W_{<n}(k)=4\sum_{j\ge t+1}U(3t-1,j),
\qquad
N(k,n)=4U(3t-1,t)=2C_t,
\]
so
\[
g_{2t}(t)=\frac12+\frac{B_t}{C_t}=1-\frac{D_t}{C_t}\in\Bigl(\frac12,1\Bigr).
\]
For $x=t-1$, set
\[
A_t:=\sum_{j=0}^{t-2}U(3t-2,j)-\sum_{j\ge t}U(3t-2,j).
\]
Using the identity $U(m+1,j)=U(m,j)+2U(m-1,j-1)$ and summing over the defining ranges gives
\[
D_t=2A_t+B_t.
\]
Since $N(3t,2t)=4U(3t-2,t-1)=C_t$, we obtain
\[
g_{2t}(t-1)=\frac12+\frac{2A_t}{C_t}
=\frac12+\frac{D_t-B_t}{C_t}
=\frac{2D_t}{C_t}\in(0,1).
\]
Therefore
\[
g_{2t}(t-1)+g_{2t}(t)=1+\frac{D_t}{C_t}>1.
\]
Using $p_n(t-1)=p_n(t)$, we conclude that
\[
\PP(I_n=1)
\ge \sum_{x=0}^{t-2}p_n(x)+p_n(t)\bigl(g_{2t}(t-1)+g_{2t}(t)\bigr)
=\frac12+p_n(t)\frac{D_t}{C_t}
>\frac12.
\]
This proves both the lower bound and strictness for every $n\ge2$.
\end{proof}

The average-length consequences are now immediate.

\begin{corollary}[Uniform benchmark improvement]\label{cor:uniform}
For every $n\ge1$,
\[
\EE[|C(X_1^n)|]\le \frac32\,n,
\]
with equality if and only if $n=1$. Equivalently,
\[
\EE[|C(X_1)|]=\frac32,
\qquad
\EE[|C(X_1^n)|]<\frac32\,n\quad (n\ge2).
\]
\end{corollary}

\begin{proof}
Equation~\eqref{eq:Ln-decomp} gives
\[
\EE[L_n]=\EE[K_n]-\PP(I_n=1)=\frac32 n+\frac12-\PP(I_n=1).
\]
Now apply \Cref{lem:key}.
\end{proof}

A slightly more careful reading of the proof already yields the first-order asymptotic gain.

\begin{corollary}[Asymptotic $n^{-1/2}$ improvement]\label{cor:asymp}
Let $L_n:=|C(X_1^n)|$. Then, as $t\to\infty$,
\[
\frac32(2t+1)-\EE[L_{2t+1}]
\ge
\Bigl(\frac{1}{18\sqrt\pi}+o(1)\Bigr)t^{-1/2}
=
\Bigl(\frac{\sqrt2}{18\sqrt\pi}+o(1)\Bigr)(2t+1)^{-1/2},
\]
and
\[
\frac32(2t)-\EE[L_{2t}]
\ge
\Bigl(\frac{1}{9\sqrt\pi}+o(1)\Bigr)t^{-1/2}
=
\Bigl(\frac{\sqrt2}{9\sqrt\pi}+o(1)\Bigr)(2t)^{-1/2}.
\]
In particular, for every fixed
\[
0<c<\frac{\sqrt2}{18\sqrt\pi}
\]
there exists $n_0(c)\in\mathbb N$ such that for all $n\ge n_0(c)$,
\[
\EE[|C(X_1^n)|]\le \frac32\,n-\frac{c}{\sqrt n}.
\]
\end{corollary}

\begin{proof}
Let
\[
q_t:=\binom{2t}{t}4^{-t}.
\]
Then
\[
p_{2t+1}(t)=q_t,
\qquad
p_{2t}(t)=\binom{2t-1}{t}2^{-(2t-1)}=q_t.
\]
By Stirling's formula,
\[
q_t=\frac{1}{\sqrt{\pi t}}\,(1+o(1)).
\]
Moreover, \Cref{lem:central} gives
\[
\frac{D_t}{C_t}
=\sum_{j=1}^{t}(-1)^{j-1}\frac{C_{t-j}}{C_t}
\longrightarrow
\sum_{j\ge1}(-1)^{j-1}8^{-j}
=\frac19,
\]
because for each fixed $j$, $C_{t-j}/C_t\to 8^{-j}$, and the summands are dominated by a geometric sequence.

In the odd case, the proof of \Cref{lem:key} gives
\[
\PP(I_{2t+1}=1)\ge \frac12+q_t\frac{D_t}{2C_t},
\]
hence
\[
\frac32(2t+1)-\EE[L_{2t+1}]
=\PP(I_{2t+1}=1)-\frac12
\ge q_t\frac{D_t}{2C_t}
=\Bigl(\frac{1}{18\sqrt\pi}+o(1)\Bigr)t^{-1/2}.
\]
In the even case,
\[
\PP(I_{2t}=1)\ge \frac12+q_t\frac{D_t}{C_t},
\]
so
\[
\frac32(2t)-\EE[L_{2t}]
=\PP(I_{2t}=1)-\frac12
\ge q_t\frac{D_t}{C_t}
=\Bigl(\frac{1}{9\sqrt\pi}+o(1)\Bigr)t^{-1/2}.
\]
The final statement follows since the odd lower bound is the smaller of the two asymptotic constants.
\end{proof}

\section*{Acknowledgements}
ChatGPT-5.4 provided the main ideas used in the proof of \Cref{lem:central}.

\end{document}